\documentclass[12pt]{article}
\newcommand{\arxiv}[2]{#1} 

\arxiv{
	\usepackage{fullpage} 
	\usepackage{amsthm} 
	\theoremstyle{plain}
	\newtheorem{theorem}{Theorem}[section]
	\newtheorem{lemma}[theorem]{Lemma}
	\newtheorem{corollary}[theorem]{Corollary}
	
	\newtheorem{example}[theorem]{Example}

	\theoremstyle{definition}
	\newtheorem*{definition}{Definition}
	
	\theoremstyle{remark}

	\usepackage{enumitem}
	\setlist[enumerate]{label*=(\arabic*),
		topsep=0.25em,
		parsep=0.1em,
		itemsep=0.1em
		}
	\setlist[itemize]{
		topsep=0.25em,
		parsep=0.1em,
		itemsep=0.1em
	}
	
	\newtheoremstyle{named}{}{}{\itshape}{}{\bfseries}{.}{.5em}{\thmnote{#3}}
	\theoremstyle{named}
	\newtheorem*{namedtheorem}{Theorem}
	\newcommand{\flexdef}[2]{
		\begin{namedtheorem}[#1]
			#2
		\end{namedtheorem}
	}
}{
	\newcommand{\flexdef}[2]{
		\begin{definition}
			#2
		\end{definition}
	}
}

\usepackage{mathpartir}

\usepackage[english]{babel}
\usepackage{ latexsym }
\usepackage{ amssymb }

\usepackage{tikz}
\usetikzlibrary{shapes}

\usepackage{amsmath}

\usepackage{thmtools}
\usepackage{mathtools} 
\usepackage{xspace} 
\usepackage{dsfont} 
\usepackage{xifthen} 
\usepackage{graphicx} 
\usepackage[subnum]{cases} 
\usepackage{multirow} 

\usepackage{hyperref}

\usepackage[utf8]{inputenc}

\usepackage{tabularx}

\usepackage{float}
\usepackage{placeins} 

\newcommand{\maximum}{maximum\xspace}

\newcommand{\textdefine}[1]{\textit{#1}}

\newcommand{\comment}[1]{}

\newcommand{\name}[1]{\textsc{#1}\xspace}

\newcommand{\problemm}[1]{\textsc{#1}\xspace}

\newcommand{\problemTSS}{\problemm{TSS}}
\newcommand{\problemTSSLong}{\problemm{Target Set Selection}}

\newcommand{\thr}{\myfunction{thr}}

\newcommand{\fpt}{\name{FPT}}

\newcommand{\np}{\name{NP}}

\newcommand{\wclass}[1]{\name{W[\( #1 \)]}}

\newcommand{\apx}{\name{APX}}

\newcommand{\lab}{\myfunction{lab}} 

\newcommand{\tmax}{t_{\max}}
\newcommand{\tamount}[1]{t^{#1}}

\newcommand{\add}[1]{\eta_{#1}}
\newcommand{\addd}{\add{\alpha,\beta}}
\newcommand{\addda}{\add{\alpha' \! , \beta'}}

\newcommand{\GG}{G}

\newcommand{\myfunction}[1]{\mbox{\normalfont{\textsf{#1}}}}

\newcommand{\vertify}{\operatorname{\pi^{\text{-}1}}} 

\newcommand{\getorder}{\operatorname{\pi}} 

\newcommand{\labset}{\myfunction{labels}} 

\newcommand{\afo}{\myfunction{afo}} 
\newcommand{\tmaxx}[1]{\tmax^{#1}}
\newcommand{\rename}[1]{\rho_{#1}}
\newcommand{\renamee}{\rename{\alpha \to \beta}}

\newcommand{\nullf}{\textbf{0}}
\newcommand{\inc}{\ast}

\newcommand{\adde}{\myfunction{e}_{\getorder}} 

\newcommand{\deact}[1]{{#1}^{\getorder}} 
\newcommand{\deactt}{\deact{A}}
\newcommand{\deacttt}[2]{{#1}^{#2}}
\newcommand{\afoa}[1]{\afo(\deactt(#1))}
\newcommand{\incoming}[2]{{#1}^{<}_{#2}}

\newcommand{\condense}{\myfunction{condense}}

\newcommand{\nexte}[1]{\overrightarrow{\getorder_{#1}}}

\newcommand{\Nbh}[1]{\myfunction{N}_{#1}}

\newcommand{\addvalue}[1]{\myfunction{add}({#1})} 

\newcommand{\exampleA}{\big( (\beta,1),\allowbreak (\alpha,1),\allowbreak (\beta,1),\allowbreak (\alpha,2),\allowbreak (\alpha,2),\allowbreak (\gamma,2),\allowbreak (\beta,1),\allowbreak (\gamma,2),\allowbreak (\gamma,2) \big)}

\newcommand{\exampleAf}{\big( (\beta,1),\allowbreak (\alpha,1),\allowbreak (\beta,1),\allowbreak (\alpha,2),\allowbreak (\alpha,2),\allowbreak (\gamma,2),\allowbreak \boldsymbol{(\gamma,2),\allowbreak (\beta,1)},\allowbreak (\gamma,2) \big)}

\newcommand{\setjoin}{\mathcal{S}[f_1 \oplus f_2,(A,\afo)]}
\newcommand{\setlabel}{\mathcal{S}[\renamee f',(A,\afo)]}

\newcommand{\N}{\mathbb{N}}

\newcommand{\backward}{(\( \Leftarrow \))\xspace}
\newcommand{\forward}{(\( \Rightarrow \))\xspace}

\newcommand{\set}[2]{ \{ #1 \; | \; #2 \} }
\newcommand{\bigset}[2]{ \big\{ #1 \; \big| \; #2 \big\} }

\newcommand{\Oh}{\mathcal{O}}
\newcommand{\define}[1]{\textnormal{#1}}  

\newcommand{\problemdefSimple}[3]{
\begin{tabularx}{\textwidth}{ r X p{0.5cm} }
& &
\\ \multicolumn{2}{l}{#1} \\
Input: & #2 \\
Question: & #3
\\ & & 
\end{tabularx}
}

\arxiv{
	\newcommand{\proofSketch}{Proof (Sketch)}
	
	\newcommand{\seeProof}[1]{
		\begin{proof}
			See appendix #1.
		\end{proof}}
	\newcommand{\seeFullProof}[1]{For a full proof see appendix #1}
}{
	\newcommand{\proofSketch}{Sketch}
	
	\newcommand{\arxivlink}{\href{https://arxiv.org/abs/0706.1234}{https://arxiv.org/abs/0706.1234}}
	\newcommand{\seeProof}[1]{}
	\newcommand{\seeFullProof}[1]{For a full proof see online version}
}
\begin{document}
	
\arxiv{
	\begin{center}
	{\LARGE{{Target Set Selection Parameterized by\\ Clique-\!Width and Maximum Threshold}}}
	
	\ \\
	
	{\small{Tim A.~Hartmann\\
	Lehrstuhl für Informatik 1,\\ RWTH Aachen University, Germany\\
		hartmann@algo.rwth-aachen.de}}
	\end{center}
}{
	\title{Target Set Selection Parameterized by Clique-\!Width and Maximum Threshold}
	\author{Tim A.~Hartmann}
	\institute{Lehrstuhl für Informatik 1, RWTH Aachen University, Germany\\
	\email{hartmann@algo.rwth-aachen.de}
	}
	\maketitle
}

\pagenumbering{arabic}
\setcounter{page}{1}

\begin{abstract}
	The \problemTSSLong problem takes as an input a graph $G$ and a non-negative integer threshold \( \thr(v) \) for every vertex $v$.
	A vertex $v$ can get active as soon as at least \( \thr(v) \) of its
	neighbors have been activated. The objective is to select a smallest
	possible initial set of vertices, the target set, whose activation eventually leads
	to the activation of all vertices in the graph.
	
	We show that \problemTSSLong is in \fpt when parameterized with	the combined parameters clique-width of the graph and the maximum threshold value.
	This generalizes all previous \fpt-membership results for the parameterization by maximum threshold,
		and thereby solves an open question from the literature.
	We stress that the time complexity of our algorithm is surprisingly well-behaved and grows only single-exponentially in the parameters. 
\end{abstract}
\section{Introduction}

The \problemTSSLong problem (\problemTSS) suits to model irreversible propagation of all sorts of conditions or information in a network.
	This may be for example a word-of-mouth-effect, disease spreading or fault influence in distributed systems \cite{DBLP:journals/snam/NichterleinNUW13}.
The input is an undirected graph \( G \) and a non-negative integer threshold \( \thr(v) \) for every vertex \( v \).
The task is to select a smallest possible set \( S \) of initially active vertices, the target set, whose activation eventually leads to the activation of all vertices in the graph.
A vertex $v$ can become active as soon as at least \( \thr(v) \) of its
	neighbors have been activated.

Our view on the activation of a vertex is that it is \textit{allowed} to become active if enough neighbors are active before,
	in contrast to that it is \textit{obligated} to get active as soon as possible.
We ask for a smallest possible set \( S \), the target set, and a permutation of the vertices \( \getorder \), which is the ordering in which the vertices get active. 
Then, for every non-target set vertex \( v \), to assure its activation we require that at least threshold \( \thr(v) \) many neighbors of \( v \) are ordered before \( v \).
In particular, our permutation may order the target set vertices \( S \) not at the beginning.
This definition is more robust towards re-orderings of the permutation of vertices.
We can re-order the permutation and not have to bother that for example the target set no longer consists of the very first vertices of the ordering.
In the literature the problem is commonly defined via rounds of activations that define sets of active vertices for each round.
Our definition is equivalent while being much more convenient for our techniques.
\problemdefSimple{\problemTSSLong}
	{An undirected graph \( G \), a non-negative threshold for every vertex \( \thr: V(G) \to \N \), and \( k \in \N \).}
	{
	Is there a set of vertics \( S \subseteq V(G) \) of size at most \( k \) and a permutation of the vertices \( \getorder: V(G) \to [|V(G)|] \) such that for every vertex \( v \in V(\GG) \setminus S \) we have 
	\( | \bigset{ u \in \Nbh{G}(v) }{ \getorder(u) < \getorder(v) } | \; \geq \; \thr(v)  \)?
	}
The problem was first introduced by Kempe et al.~\cite{DBLP:conf/kdd/KempeKT03}.
It proves to be computationally extremely difficult.
It is \np-hard even for the restriction to split-graphs of diameter two \cite{DBLP:journals/snam/NichterleinNUW13}.
		Chen showed that minimizing the size of the target set is \apx-hard \cite{DBLP:journals/siamdm/Chen09}. 
		More recently, Bazgan et al.~showed that for every functions \( f \) and \( \rho \) this problem cannot be approximated within a factor of \( \rho(k) \) in \( f(k) \cdot n^{\Oh(1)} \) time \cite{DBLP:journals/computability/BazganCNS14}. 
The parameterized complexity studies focus on the original problem and two variants that limit the allowed thresholds.
	These are \textit{constant thresholds}, where all thresholds are at most a constant \( \tmax \),
	and \textit{majority thresholds}, where a vertex can get active as soon as at least the majority of its neighborhood is active before.
The general \problemTSS is \wclass{1}-hard for each of the parameterization, ``distance to cluster,'' \cite{DBLP:journals/mst/ChopinNNW14} ``distance to forest'' and pathwidth \cite{DBLP:journals/snam/NichterleinNUW13}.
The strongest positive \fpt-membership results for constant thresholds are the parameterization by treewidth \cite{DBLP:journals/disopt/Ben-ZwiHLN11},
	the parameterization by ``distance to cluster'' \cite{DBLP:journals/mst/ChopinNNW14},
	and the parameterization by neighborhood diversity \cite{TSSinDenseGraphClasses2016}.
There are a lot more parameterized complexity results for these three variants of \problemTSS \cite{DBLP:journals/mst/ChopinNNW14,DBLP:journals/snam/NichterleinNUW13}.
%
Further, Cicalese et al.~study a variant of \problemTSS which asks if a set of vertices \( A \) can be activated in a given number of activation rounds \cite{DBLP:journals/tcs/CicaleseCGMV14}.
They give a polynomial time algorithm when the number of activation rounds and the clique-width of the input graph are constant.
Their exponential dependency on the clique-width is unlikely to be improved, as even \problemTSS for one activation round is \( \wclass{1} \)-hard with respect to the treewidth \cite{simple}.
%
%
%
%
For a more extend introduction to the history of the problem as well as other algorithmic aspects and similar models see for example \cite{DBLP:journals/mst/ChopinNNW14,DBLP:journals/snam/NichterleinNUW13}. 

Dvo\v{r}\'{a}k et al.~raised the question of the complexity of the parameterization by the modular-width \cite{DBLP:journals/corr/DvorakKT16}.
The structural graph parameter modular-width was introduced by Gajarsk{\'{y}} et al.~\cite{DBLP:conf/iwpec/GajarskyLO13}.
We give a positive answer by showing \fpt-membership for a more general question.
We consider the clique-width which is upper bounded by the parameters modular-width and treewidth \cite{DBLP:journals/siamcomp/CorneilR05}, and by further common structural parameters for which the parametrized complexity of \problemTSS was open. 
Thereby, we generalize all positive \fpt-memberships results for \problemTSS with constant thresholds.
Further, our result does not rely on the \maximum threshold \( \tmax \) being a constant, but allows that \( \tmax \) is a parameter.
Moreover, the time complexity of our algorithm behaves surprisingly well and grows only single-exponentially in the parameters clique-width and \maximum threshold.

A related result is that \problemTSS is in \fpt when parameterized by treewidth and \maximum threshold, by Ben{-}Zwi et al.~\cite{DBLP:journals/disopt/Ben-ZwiHLN11}.
They use a dynamic program that works along the bags of a computed tree decomposition.
They fix the local ordering in which the vertices of the currently observed bag get active.
Our approach also uses such an recursive approach, while working on a computed \( \ell \)-expression.
Informally, an \( \ell \)-expression is a tree-decomposition in the context of clique-width.
Such an \( \ell \)-expression \( f \) uses three types of recursive operations that work on labeled vertices using at most \( \ell \) different labels.
Analogously to the approach for a tree decomposition, for every subexpression a current state fixes a part of the global ordering of the vertices.

However, the described vertices of a current subexpression is not bounded by our parameters.
Our algorithm has to remember an ordering of a limited number of vertices and further has to address these vertices indirectly.
Crucial for the activation of a vertex is its threshold and neighborhood.
	However, we cannot address the neighborhood even for vertices of currently equal label and threshold since they can have very different neighborhoods as subexpression may reveal.
Consequently, our approach explores the \( \ell \)-expression top down, and fixes an ordering of the important vertices of the up to now described graph.
The up to now encountered operations define a common neighborhood for all vertices of a fixed label.
	This is because for every outer operations, vertices of the same label behave equally.
Thus, our local ordering indirectly references the vertices solely by their label and threshold.

Further, vertices of the same label that occur late enough in a global ordering behave equally.
There is only one type of edge operation of \( \ell \)-expression, namely \( \addd \) adding all edges between vertices of some labels \( \alpha \) and \( \beta \).
There, for a vertex \( v \) of label \( \alpha \) we have to account the contribution to the activation of \( v \) due to vertices of label \( \beta \).
Only the first \( \thr(v) \leq \tmax \) active vertices of label \( \beta \) are important.
If the activation of \( v \) is between the activation of the first \( \tmax \) of label \( \beta \), we fix their relative positioning in our local ordering.
Otherwise, the activation  of \( v \) does not differ from other late vertices of label \( \alpha \).

However, we need to guarantee that a vertex \( v \) of label \( \alpha \) that is not referenced by our local ordering is indeed ordered late enough.
That is, the first \( \tmax \) vertices of label \( \beta \) occur before vertex \( v \).
We denote such a global ordering as nice to the current subexpression \( \addd f' \).
It is possible to modify any valid global ordering to be nice to all subexpressions.
We extend our local ordering to also include the \( (\tmax+1) \)-st vertex of every label.
	Then, whether the underlying global ordering \( \getorder \) is nice, is reflected in our local ordering.
Therefore, we can restrict our algorithm to consider nice global orderings only.

The resulting procedure for our algorithm at each operation of the given \( \ell \)-expression then is as follows.
For a current edge operation \( \addd \),
	for each vertex \( v \) we simply have to adjust the number of neighbors contributing to the activation of \( v \) according to our fixed local ordering.
We remember this contribution as the \textdefine{activation from outside}.
For a current operation that combines two subgraphs, consider the unknown partition of the vertices fixed by the local ordering in either subgraph.
	In that case, the algorithm tries all possibilities.
The approach for the operation that re-labels a label is very similar.
For every subexpression, the number of possible states is single-exponentially bounded by our parameters, which yields to an overall \fpt-runtime.

\newcommand{\lemmaMain}{
	Let \( \tmax, \ell \in \N \).
	There is an algorithm that,
		given a graph \( G \), a threshold for each vertex \( \thr: V(G) \to [0, \tmax] \) and an \( \ell \)-expression \( f \) of \( G \),
		computes the minimal size of a target set in time \( \Oh( \ell^{3 \ell t} \cdot t^{\ell (4t+1) } \cdot |f| ) \), where \( t := \tmax+1 \) and \( |f| \) is the length of \( f \).
	}
\begin{theorem}
\label{lemma:main:introduction}
	\lemmaMain
\end{theorem}

An easy upper bound for the length of the \( \ell \)-expression \( f \) is \( |V(G)|^2 \). 
Further, one can obtain a minimum target set, and not only its size, by tracking such sets throughout our dynamic program.

Oum gave an algorithm that either outputs an \( (8^\ell-1 ) \)-expression of graph \( G \) or confirms that the clique-width of \( G \) is larger than \( \ell \), and that runs in time \( \Oh( g(\ell) \cdot |V(G)|^3 ) \), where \( g(\ell) \) only depends on the clique-width \( \ell \) \cite{DBLP:journals/talg/Oum08}.
Combined with the algorithm of Theorem \ref{lemma:main:introduction} it follows that \problemTSS parameterized by the clique-width and the \maximum threshold is in \fpt.

\begin{corollary}
	\problemTSSLong is in \fpt with respect to the combined parameters clique-width of the given graph and the \maximum threshold.
\end{corollary}

Following the preliminaries in \autoref{section:preliminaries}, we prove Theorem \ref{lemma:main:introduction} in \autoref{section:dynamic:program}.
We conclude in \autoref{section:conclusion}. %
\arxiv{}{
	Due to space constraints, we omit some proofs or only give a proof sketch.
	For the full proof, we refer to an online version at \arxivlink.
}


\section{Preliminaries}
\label{section:preliminaries}

For integers \( i < j \), let \( [i] := \{1,2,\dots,i\} \) and \( [i,j] := \{i,(i\!+\!1),\dots,j \} \).
For a list (or vector) \( A \), we describe the \( i \)-th element as \( A[i] \).

All our graphs are simple, finite and undirected.
For a graph \( G \), we denote by \( V(G) \) its set of vertices.
We use \( \Nbh{\GG}(v) \) as the neighborhood of vertex \( v \in V(\GG) \).
Usually we consider graphs with thresholds for each vertex \( \thr: V(\GG) \to [0,\tmax] \) which are at most a constant \( \tmax \),
	and assume that its thresholds \( \thr \) and \( \tmax \) are given, if needed.

In this work, we consider parameterized complexity.
For an introduction see for example \cite{DBLP:books/sp/CyganFKLMPPS15,DBLP:journals/corr/abs-1106-3161,groheBuch,NiedermeierInvitationToFixedParameterAlgorithms2006}.
For a graph class, for example clusters (the disjoint union of cliques), the parameter ``distance to cluster'' is the minimal number of vertices one needs to delete from the input graph in order to obtain a cluster.

The clique-width \( \myfunction{cw}(G) \) of a graph \( \GG \) was introduced in \cite{DBLP:journals/dam/CourcelleO00}.
A graph has clique-width at most \( \ell \in \N \), if it can be constructed by an \( \ell \)-expression that uses four types of operations
	and a labeling of the vertices of at most \( \ell \) labels, as we describe in the following.
Let \( \labset(f) \) be the set of labels used by \( f \).
To avoid confusion with thresholds, we use small Greek letters \( \alpha, \beta, \gamma \) for the labels.
An \( \ell \)-expression defines a graph \( \GG(f) \) with labels per vertex \( \lab_G: V(\GG) \to \labset(f) \).
The graph \( \GG(f) \) is recursively defined as
\begin{itemize}
	\item \( \GG(v(\alpha)) \), a single vertex \( v \) of label \( \alpha \in \labset(f) \),
	\item \( \GG(f_1 \oplus f_2) \), the disjoint union of \( \GG(f_1) \) and \( \GG(f_2) \) for \( \ell \)-expressions \( f_1 \), \( f_2 \),
	\item \( \GG( \eta_{\alpha, \beta} f' ) \), the graph \( \GG(f') \) where there is an \textbf{e}dge between every vertex of label \( \alpha \) and every vertex of label \( \beta \), for \( \ell \)-expression \( f' \), and
	\item \( \GG(\rho_{\alpha \to \beta} f') \), the graph \( \GG(f') \) where all vertices of label \( \alpha \) are \textbf{r}e-labeled to label \( \beta \), for \( \ell \)-expression \( f' \).
\end{itemize}
The subexpressions of \( f \) are all expressions \( f_1, f_2, f' \) used in the recursive definition of \( f \). Especially \( f \) is a subexpression of \( f \).
We drop the \( \GG( \cdot ) \) when using \( \GG(f) \) as a nested term.
	For example, instead of \( V(\GG(f)) \), we simply write \( V(f) \).
Further, we also refrain from specifying the set of labels \( \labset(f) \) if it is clear from the context.

An \( \ell \)-expression is \define{irredundant} if for every subexpression \( \addd f' \) the graph \( \GG(\addd f' ) \) has no edge between vertices of label \( \alpha \) and \( \beta \).
We assume that the given \( \ell \)-expression is irredundant, which we can assure by a simple preprocessing step \cite{DBLP:journals/dam/CourcelleO00}.

\section{Dynamic Program}
\label{section:dynamic:program}

A good way to convince someone that a graph \( \GG \) with thresholds has a target set of size at most \( k \) is to state a complete ordering in which the vertices get active.
We denote this permutation of the vertices as a \textdefine{global ordering} \( \getorder: V(\GG) \to [|V(\GG)|] \).
We say that \( \getorder \) is \( k \)-activating for graph \( \GG \) if
	there is a \( k \)-vertex set \( S \subseteq V(\GG) \), the target set,
	such that for every other vertex \( v \) the neighbors of \( v \) that are ordered before \( v \) outnumber the threshold \( \thr(v) \).

\flexdef{Global Ordering}{
\label{example:graph}
	A \define{global ordering} of a graph \( \GG \) is a permutation of the vertices \( \getorder: V(\GG) \to [|V(\GG)|] \).
	Further, \( \getorder \) is \( k \)-activating (for \( \GG \)) if
		there is a \( k \)-vertex set \( S \subseteq V(G) \) such that for every vertex \( v \in V(\GG) \setminus S \) we have
	\[ \incoming{\getorder}{\GG}(v) \; := \; \big| \bigset{ u \in \Nbh{\GG}(v) }{ \getorder(u) < \getorder(v) } \big| \; \geq \; \thr(v) . \]
	Graph \( G \) has a target set of size \( k \) if
		there is a global ordering \( \getorder \) such that \( \getorder \) is \( k \)-activating for \( \GG \).
}

\begin{example}
	The following graph \( G \) has global ordering \( \getorder: v_i \mapsto i \), 
		which is 1-activating (for \( S = \{ v_1 \} \)).
	Further, \( f = \add{\beta,\gamma} f' =
	\add{\beta,\gamma} ( v_6(\gamma) \oplus v_8(\gamma) \oplus v_{11}(\gamma) \oplus v_9(\gamma) \oplus v_7(\beta) \oplus \rename{\gamma \to \alpha} \add{\beta,\gamma} ( v_{10}(\gamma)
		\oplus \rename{\gamma \to \alpha} \add{\alpha, \beta} \add{\alpha,\gamma} \add{\beta, \gamma} (  
			v_2(\gamma) \oplus v_1(\beta) \oplus v_3(\beta) \oplus v_4(\alpha) \oplus v_5(\alpha)
		 ) ))
	\)
	is a \( 3 \)-expression of \( G \).
	For each vertex, the label among \( \{\alpha,\beta,\gamma\} \) and threshold at most \( \tmax = 2 \) is given as a tuple.
	\\ \begin{center}\newcommand{\posXA}{0}
\newcommand{\posXAB}{2.5}
\newcommand{\posXB}{5}
\newcommand{\posXC}{10}

\newcommand{\posYZ}{-1.2}
\newcommand{\posYB}{1.2}
\newcommand{\posYC}{2.4}

\begin{tikzpicture}
[scale=0.7,auto=left, node/.style={circle,fill=white, draw, scale=0.5}
	,max/.style={circle,fill=black, draw, scale=0.5}]

	\node[node] (a1) at (\posXA,0) {};
	\node[node] (a2) at (\posXA,\posYB) {};
	\node[node] (a3) at (\posXA,\posYZ) {};
	
	\node[node] (bb) at (\posXAB,\posYZ) {};
	
	\node[max] (b1) at (\posXB,0) {};
	\node[node] (b2) at (\posXB,\posYB) {};
	\node[node] (b3) at (\posXB,\posYZ) {};

	\node[node] (c1) at (\posXC,\posYB) {};
	\node[node] (c2) at (\posXC,0.4) {};
	\node[node] (c4) at (\posXC,-0.4) {};
	\node[node] (c3) at (\posXC,\posYZ) {};
	
	\node [below] at (b1.south) {\( \; \; (\beta, 1) \)};
	\node [above] at (b1.north) {\( v_1 \)};
	\node [below] at (b3.south) {\( \; v_7 \), \( (\beta, 1) \)};
	\node [above] at (b2.north) {\( v_3 \), \( (\beta, 1) \)};
	
	\node [below] at (bb.south) {\( v_2 \), \( (\alpha, 1) \; \)};	
	\node [left] at (a1.west) {\( v_4 \), \( (\alpha, 2) \)};
	\node [left] at (a2.west) {\( v_5 \), \( (\alpha, 2) \)};
	\node [left] at (a3.west) {\( v_{10} \), \( (\alpha, 2) \)};
	
	\node [right] at (c1.east) {\( v_6 \), \( (\gamma,2) \)};
	\node [right] at (c2.east) {\( v_8 \), \( (\gamma,2) \)};
	\node [right] at (c4.east) {\( v_{11} \), \( (\gamma,2) \)};
	\node [right] at (c3.east) {\( v_9 \), \( (\gamma,2) \)};
	
	\foreach \to in {a1, a2, b1, b2}
		\draw (bb) -- (\to); 
	
	\foreach \from/\to in {a1/b1,a2/b2,a1/b2, a3/b1, a3/b2, a2/b1}
		\draw (\from) -- (\to); 
	
	\foreach \b in {1,...,3}{
		\foreach \c in {1,...,4}{
			\draw (b\b) -- (c\c);
		}
	}
\end{tikzpicture}\end{center}
	For later examples, let \( A := \exampleA \), and further \( \addd f' \), \( G \) and \( \getorder \) be as defined here. 
\end{example}
\newcommand{\lazyexample}[1]{}

An \( \ell \)-expression \( f \) describes a graph \( G(f) \) with three types of recursive operations that rely on \( \ell \) different labels assigned to the vertices.
We formulate a dynamic program over the subexpressions of \( f \).
At a current subexpression \( f \), a \textdefine{state} fixes a part of a global ordering \( \getorder \).
Whether such a state is a part of a \( k \)-activating global ordering, is verified by considering the subexpressions with suitable states.

In order to obtain the desired \fpt-runtime,
	we may only work with states that fix an ordering of a number of vertices bounded by our parameters, which are \maximum threshold \( \tmax \) and clique-width \( \ell \).
However, the number of all vertices described by a current subexpression is not bounded by our parameters.
Our algorithm thus can only remember an ordering of a limited number of vertices and further cannot address these vertices directly.
We identify the important verices and a suitable way to remember them.
Crucial for the activation of a vertex is its threshold and neighborhood.
	Our local ordering can very well remember the threshold of vertices.
	However, it cannot address the neighborhood even for vertices of currently equal label and threshold since they can have very different neighborhoods as subexpression may reveal.

Consequently, our approach explores the given \( \ell \)-expression top down, and fixes an ordering of the important vertices of the graph described by the up to now seen part of the \( \ell \)-expression.
The up to now seen operations define a common neighborhood for all vertices of a fixed label.
	This is because for every outer operation, two vertices of equal label behave equally.
Thus, our local ordering can indirectly reference the vertices solely by their label and threshold.

Now, let us identify the vertices whose relative ordering is crucial.
We can observe that vertices of the same label that occur late enough in a global ordering behave equally.
An \( \ell \)-expression has only one type of operation that adds edges, namely \( \addd \) for some labels \( \alpha \) and \( \beta \),
	which adds all edges between vertices of labels \( \alpha \) and \( \beta \).
There, for a vertex \( v \) of label \( \alpha \) we have to account for the contribution to the activation of \( v \) by the vertices of label \( \beta \).
Only the first \( \thr(v) \leq \tmax \) vertices of label \( \beta \) of the global ordering \( \getorder \) are important.
Consequently, if \( \getorder \) orders \( v \) somewhere between the first \( \tmax \) vertices of label \( \beta \),
	the local ordering fixes the ordering of \( v \) relatively to those first vertices of label \( \beta \) as well.
If \( \getorder \) orders \( v \) after the first \( \tmax \) of label \( \beta \), we can neglect its exact ordering.
	This is because the number of neighbors of label \( \beta \) that contribute to its activation do not differ from other such late vertices of label \( \alpha \).
Our plan therefore is that the local ordering fixes the relative positioning of these crucial first \( \tmax \) vertices of every label.

Doing so, we need to guarantee that a vertex \( v \) of label \( \alpha \) that is not referenced by our local ordering is indeed ordered late enough.
That is, the first \( \tmax \) vertices of label \( \beta \) occur before vertex \( v \).
	In particular, the first \( \tmax \) vertices of label \( \beta \) are ordered before the \( (\tmax+1) \)-st of label \( \alpha \).
	Then, given that \( v \) is not referenced by our local ordering, there are at least \( \tmax \) of label \( \beta \) ordered before,
		or if there are not even as many of label \( \beta \), accordingly less.
We denote such an ordering as nice to the current subexpression \( \addd f' \).
It is possible to modify any valid global ordering such that it is nice to every subexpression.
Therefore, our algorithm may only consider nice global orderings.
We extend our local ordering to also include the \( (\tmax+1) \)-st vertex of every label.
	Then, whether the underlying global ordering \( \getorder \) is nice to a current expression \( \addd f' \), is reflected in our local ordering.
Our algorithm may then ignore states with such not nice local orderings.

We define the local ordering \( A \) for a current \( \ell \)-expression \( f \) that fixes the relative ordering of the first \( (\tmax+1) \) activate vertices for each label \( \alpha \)
	(or if there are not even as many vertices of label \( \alpha \), accordingly less), which we denote by \( \tamount{\alpha} \).
We indirectly remember a vertex \( v \) by fixing the label and threshold of \( v \).
For technical reasons, we define a local ordering as possibly incomplete.
Our algorithm only considers complete local orderings.

\flexdef{Local Ordering}{
Let \( \GG \) be a graph with labels \( \lab: V(\GG) \to \labset(\GG) \).
For label \( \alpha \), let \( \tamount{\alpha}(\GG) :=\allowbreak \min\{ \allowbreak \tmax( \GG )+1,\allowbreak \; | \set{v \in V(\GG)}{\lab(v) = \alpha} | \} \).
A \define{local ordering} \( A \) of \( \GG \) is a list of tuples of label and threshold \( (\alpha, a) \in \labset(\GG) \times [0,\tmax(\GG)] \) such that for every label \( \alpha \) there are at most \( \tamount{\alpha} \) tuples of label \( \alpha \); and
\( A \) is \define{complete} if, for every label \( \alpha \), there are exactly \( \tamount{\alpha} \) tuples of label \( \alpha \).
}

The local ordering \( A \) is our limited view on a global ordering \( \getorder \).
Let \( \condense(\getorder) \) be the ordered list of vertices consisting of the first \( \tamount{\alpha} \) vertices of each label \( \alpha \).
A global ordering \( \getorder \) \textdefine{extends} \( A \) if the tuples of label and threshold of \( \condense(\getorder) \) are equal to \( A \).
%
As a technical tool, we also define \( \condense(\getorder,A) \) as the first ordered vertices consisting of each label \( \alpha \),
 	such that the number of vertices labeled \( \alpha \) is equal to as there are in \( A \).

\flexdef{Extending a Local Ordering}{
	Let graph \( \GG \) have global ordering \( \getorder \).
	Consider the list of vertices according to the global ordering \( \vertify(1), \dots,\allowbreak \vertify(|V(\GG)|) \).
	For every label \( \alpha \), remove all vertices of label \( \alpha \) but the first \( \tamount{\alpha} \) vertices of label \( \alpha \).
	Then, the resulting list is \( \condense(\getorder) \).
	Global ordering \( \getorder \) \define{extends} a local ordering \( A \) (for \( \GG \)) if the list tuples of label and threshold of \( \condense(\getorder) \) is equal to \( A \).
	
	Let \( \condense(\getorder, A) \) be the remaining list, after,
			for every label \( \alpha \), removing all vertices of label \( \alpha \) but the first \( |\set{ i }{ \lab(A[i]) = \alpha }| \) of label \( \alpha \).
}

\begin{example}
	\lazyexample{Consider the graph, global ordering \( \getorder \) and local ordering \( A \) of Example \ref{example:graph}.}
	We have \( \tamount{\alpha}, \tamount{\beta}, \tamount{\gamma} = 3 \)
		and \( A \) is a complete local ordering of \( \GG \). 
	Further, \( \condense(\getorder) = \condense(\getorder,A) = (v_1, \dots ,v_9) \), 
		whose list of tuples of label and threshold is equal to \( A \).
	Thus, \( A \) extends \( \getorder \).
	Let incomplete local ordering \( A^\inc \) contain only one tuple per label.
	Then, \( \condense(\getorder, A^\inc) \) is the list of vertices \( (v_1, v_2, v_6) \).
	The list of tuples of label and threshold is equal to \( A^\inc \).
\end{example}

For an edge operation \( \addd \), which adds all edges between vertices of two distinct labels,
	we simply have to adjust the number of neighbors contributing to an activation of a vertex according to our fixed local ordering.
We remember this contribution as the \textdefine{activation from outside}.
%
The mapping \( \afo \) maps to a value \( [0,\tmax] \) for each position of the local ordering \( A \), as well as maps to a value for each label.
That way we have a value for every vertex indirectly referenced by \( A \).
Further, there is a value for every vertex \( v \) not referenced by \( A \), which we identify via the label of \( v \).

A \textdefine{state} of a current subgraph \( \GG(f) \) is a tuple consisting of a local ordering \( A \) and an activation from outside \( \afo \).
To reference the activation from outside for a concrete vertex \( v \) we define \( \deactt(v) \) such that \( \afo(\deactt(v) ) \) is the activation from outside for \( v \).
Thus, \( \deactt(v) \) maps \( v \) to its according position in \( A \) if it exists and otherwise to the label of \( v \).
A global ordering \( \getorder \) is \( k \)-activating for a state \( (A,\afo) \) of \( \GG \) if it is \( k \)-activating for \( \GG \)
	while supported by the activation from outside \( \afo \).

\flexdef{Activation From Outside}{
	Let \( f \) be an \( \ell \)-expression, and graph \( \GG(f) \) have local ordering \( A \).
	An \define{\textbf{a}ctivation \textbf{f}rom \textbf{o}utside} for \( A \) is a mapping \( \afo: [|A|] \cup \labset(f) \to [0,\tmax] \).
	Then, the tuple \( (A,\afo) \) is a \define{state} of \( \GG(f) \).
	For a global ordering \( \getorder \) of \( \GG(f) \), let
	\( \deactt : V(f) \to [|A|] \cup \labset(f) \),
	\[ \deactt(v) \mapsto
		\begin{dcases}
			i,			& i \in [|A|], \; v = \condense(\getorder,A)[i],
		\\	\lab(v),	& \text{else.}
		\end{dcases}
	\]
	A global ordering \( \getorder \) of \( \GG(f) \) is \define{\( k \)-activating} for \( (A,\afo) \)
		if there is \( k \)-vertex set \( S \subseteq V(G) \) such that for every vertex \( v \in V(\GG) \setminus S \) 
	 we have that
		\[ \incoming{\getorder}{\GG}(v) \; \geq \; \thr(v) - \afoa{v} . \]
}
\begin{example}
\label{example:afo}
	\lazyexample{Consider the graph, global ordering \( \getorder \) and local ordering \( A \) from Example \ref{example:graph}.}
	Let \( \afo(1) = 1 \), and for \( x \in \{2,\dots,6,\allowbreak \alpha,\allowbreak \beta, \gamma\} \), let \( \afo(x) = 0 \).
	The activation from outside for vertex \( v_1 \) is \( \afo( \deactt( v_1 ) ) = \afo(1) = 1 \) and for vertex \( v_{10} \) it is \( \afo( \deactt(v_{10} )) = \afo( \alpha ) = 0 \).
	Further, \( \getorder \) is \( 0 \)-activating for state \( (A,\afo) \).
\end{example}

\comment{
We would like to uniformly remember the activation from outside for vertices of label \( \alpha \) that are not among the first \( \tmax \) of label \( \alpha \).
Consider for example that we add edges between all vertices of label \( \alpha \) and all of label \( \beta \).
	Let there be a vertex \( v \) of label \( \alpha \) which is the \( (\tmax\!+\!2) \)-nd active of label \( \alpha \).
An unfortunate local ordering may order \( v \) before \( \tmax \) many of label \( \beta \) occur.
Our algorithm, only references the first \( (\tmax+1) \) of label \( \alpha \), and does not know how \( v \) is ordered in \( \getorder \).
We enforce that the first \( \tmax \) vertices of label \( \beta \) get active before \( v \) and thus contribute to its activation.
	In that case, an activation of the \( (\tmax+2) \)-nd vertex of label \( \beta \), and all later vertices of label \( \beta \), is most likely.
We denote such a global ordering as \textit{nice} to \( f \).
That means, every vertex of label \( \alpha \) that is not referenced by \( A \), is ordered after the first \( \tmax \) of label \( \beta \),
	and vice versa for switched \( \alpha \) and \( \beta \).
}

We define nice orderings, analogously for global orderings \( \getorder \) and local orderings \( A \).
As we show in the following, for every \( k \)-activating global ordering \( \getorder \) there is a slightly modified \( k \)-activating global ordering \( \getorder \) which is nice to every subexpression of \( f \).
Our local ordering \( A \) includes the \( (\tmax + 1) \)-st vertex of every label.
Thus, whether \( \getorder \) is to nice the current expression \( f \) is expressed in the ordering of \( A \).
Therefore, our algorithm can avoid not nice global orderings by ignoring states where the local ordering \( A \) is not nice to \( f \).

\flexdef{Nice Orderings}{
	Let \( \GG \) be a graph with global ordering \( \getorder \).
	Let \( f \) be an \( \ell \)-expression describing a subgraph of \( \GG \).
	For label \( \alpha \), let \( v_\alpha[1], v_\alpha[2], \dots \in V(f) \) be the vertices of label \( \alpha \) of \( \GG(f) \) ordered ascending according to \( \getorder \).
	For every label \( \alpha \), let \( \tmaxx{\alpha} :=\allowbreak \min\{ \allowbreak \tmax( \GG ),\allowbreak \; | \set{v \in V(\GG)}{\lab(v) = \alpha} | \} \).
	Then, \( \getorder \) is \define{nice to \( f \)} if \( f = \addd f' \) implies that (if those respective positions exist)
	\[
		\getorder( {v_\alpha}[\tmax \! + \! 1] ) \; > \; \getorder( {v_{\beta}}[\tmaxx{\beta}] )
		\; \; \; \text{ and } \; \; \; 
		\getorder( {v_\beta}[\tmax \! + \! 1] ) \; > \; \getorder( {v_{\alpha}}[\tmaxx{\alpha}] )
	. \]
	Let \( A \) be the list of tuples of label and threshold of \( \condense(\getorder \! \upharpoonright_{V(f)} ) \) for graph \( \GG(f) \),
		where \( \getorder \! \upharpoonright_{V(f)} \) is \( \getorder \) restricted to vertices \( V(f) \).
	Then, \( A \) is nice to \( f \) if (and only if) \( \getorder \) is nice to \( f \).
}

\begin{example}
\label{example:getorder:reorder}
	\lazyexample{Consider the graph \( G \), global ordering \( \getorder \) and local ordering \( A \) from Example \ref{example:graph}.}
	Global ordering \( \getorder \) is not nice to \( \add{\beta,\gamma} f' \) since
		\( \getorder( {v_\beta}[\tmax+1] ) = \getorder( v_7) = 7 \ngtr 8 = \getorder(v_8) = \getorder( {v_{\gamma}}[\tmax] ) \).
	By switching the 7th and 8th position \( \getorder \) becomes nice to \( \add{\beta,\gamma} f' \).
	Likewise, \( A \) is not nice to \( \add{\beta,\gamma} f' \), but \( A' = \exampleAf \) is nice to \( \add{\beta,\gamma} f' \).
\end{example}

\newcommand{\lemmaFix}{
	Let \( f \) be an \( \ell \)-expression and \( \getorder \) a global ordering that is \( k \)-activating for graph \( \GG(f) \).
	Then, there is a global ordering \( \getorder' \) that is \( k \)-activating for graph \( \GG(f) \)
		and nice to every subexpression of \( f \).
}
\begin{lemma}
\label{lemma:fix}
	\lemmaFix
\end{lemma}
\begin{proof}[\proofSketch]
	There may be subexpressions \( \addd f' \) where the \( (\tmax  +  1) \)st vertex of label \( \alpha \) is ordered before the first \( \tmax \) vertices of label \( \beta \), formally \( \getorder( v_\alpha[\tmax  +  1] ) =: i < \getorder( v_\beta[\tmaxx{\beta}]) \).
	We repair such a violation by moving all vertices of \( v_\beta[1], \dots, v_\beta[\tmax^\beta] \) that did not occur already between positions \( (i - 1) \) and \( i \).
	Since there are \( \tmax \) vertices of label \( \alpha \) ordered before position \( i \), the modified local ordering is still activating.
	We repair all such violations top-down.
	Following this order prevents recursive violations for already fixed subexpression \( \addda \).
	\seeFullProof{\ref{appendix:lemma:fix}}.
\end{proof}

\flexdef{\( k \)-activating}{
	Graph \( \GG(f) \) is \define{\( k \)-activating} for a state \( (A,\afo) \)
		if there is a global ordering \( \getorder \) that extends \( A \), is \( k \)-activating for \( (A,\afo) \), and is nice to every subexpression of \( f \).
}

\newcommand{\lemmaTranslateToState}{
	Let \( f \) be an \( \ell \)-expression.
	Then, graph \( \GG(f) \) has a target set of size \( k \)
		if and only if
		there is a complete local ordering \( A \) of \( \GG(f) \) such that
		\( \GG(f) \) is \( k \)-activating for state \( (A,\nullf) \), where \( \nullf: [|A|] \cup \labset(\GG) \to \{0\} \).
}

\begin{lemma}
\label{lemma:translate:to:state}
	\lemmaTranslateToState
\end{lemma}
\begin{proof}[\proofSketch]
	Use Lemma \ref{lemma:fix}.
	\seeFullProof{\ref{appendix:lemma:translate:to:state}}.
\end{proof}

It remains to specify the recursive dependency of our computation.
We distinguish the three operations, which are adding edges if \( f = \addd f' \),
	taking the disjoint union if \( f = f_1 \oplus f_2 \),
	and re-labeling if \( f = \renamee f' \).

Consider a current \( \ell \)-expression \( \addd f' \) and a state \( (A,\afo) \).
The operation \( \addd \) adds the edges between all vertices of label \( \alpha \) and \( \beta \).
We adjust the activation from outside such that it replaces the edges between vertices of label \( \alpha \) and \( \beta \).
The relative ordering of the first \( (\tmax+1) \) vertices of label \( \alpha \) and label \( \beta \) is already fixed by the local ordering \( A \).
We increase the activation from outside of a position \( y \) of \( A \) of label \( \beta \) for every prior position \( x \) of \( A \) of label \( \alpha \).
For the activation from outside for vertex \( v \) of label \( \alpha \) that is not referenced by \( A \),
	every position \( x \) of \( A \) of label \( \beta \) increases the activation from outside.
We denote the result as \( \addd \afo \).

\flexdef{\( \addd \afo \)}{
	Let graph \( \GG \) with labels \( \alpha \) and \( \beta \) have local ordering \( A \).
	For \( y \in [|A|] \cup \labset(\GG) \), let
	\begin{align*}		
		& (\addd \afo) (y) :=
			\min\!\big\{ \tmax, \; \afo( y ) + \addvalue{y}  \big\}, \; \; \text{where}
		\\& \addvalue{y} :=  |\bigset{ x \in [|A|] }{ \; x<y, \; \{\lab(x),\lab(y)\} = \{\alpha,\beta\} }|
			,
	\end{align*}
	where \( 1 <  2 < \dots < |A| < \gamma \), for every label \( \gamma \); and where 
	\( \lab(x) \), for \( x \in [|A|] \), is defined as \( \lab(A[x]) \).
	For every vertex \( v \in V(\GG) \), let
	\[ \adde(v) := | \bigset{ u \in V(\GG) }{ \getorder(u) < \getorder(v), \; \{ \lab(u), \lab(v) \} = \{\alpha,\beta\} } |. \]
}

The number of edges that additionally contribute to the activation of a vertex \( v \), denoted by \( \adde(v) \), is equal to the increase of the activation from outside \( \addvalue{v} \)
	(while ignoring an overall activation exceeding \( \tmax \)).

\newcommand{\lemmaKey}{
	Let global ordering \( \getorder \) extend local ordering \( A \), which is nice to \( \addd f' \).
	For every vertex \( v \in V(\addd f') \), we have that
		\[ \min\{ \tmax, \allowbreak \; \afo(\deactt(v)) + \adde(v) \} = (\addd \afo) (\deactt(v)) . \]
}	
\begin{lemma}
\label{lemma:key}
	\lemmaKey
\end{lemma}
\begin{proof}[\proofSketch]
	We need to show for every vertex \( v \) that the number of new neighbors ordered before, \( \adde(v) \), is equal to how much we increase \( \afo(\deactt(v)) \), when capped by \( \tmax \).
	Since \( A \) is nice to \( \addd f' \), this number of new neighbors is correctly expressed by comparing \( v \) with its neighbors of label \( \beta \) in \( A \), which is how \( \addvalue{ \deactt(v) } \) is computed.
	\comment{
	Consider the statement for a vertex \( v \) of label \( \alpha \).
	That simplifies the terms \( \adde(v) \) and \( \addvalue{ \deactt(v) } \) to 
	\( 
		\adde(v) = | \set{ u \in V }{ \getorder(u) < \getorder(v), \; \lab(u)= \beta } | \), and \(
		\addvalue{ \deactt(v) } = |\set{ x \in [|A|] }{ x<\deactt(v), \; \lab(x) = \beta }|
		.
	\) 
	We show that \( \min\{ \tmax, \; \adde(v) \} = \addvalue{ \deactt(v) } \), which implies the statement.
	
	Consider that \( v \) is among the first \( \tmax \) vertices of label \( \alpha \) according to the global ordering \( \getorder \).
	The first \( \tmaxx{\beta} \) vertices of label \( \beta \) according to \( \getorder \)
		have their relative order preserved in their representation of the local ordering \( A \).
	Instead of comparing their ordering in \( \getorder \) we can compare their ordering in \( A \), and vice versa.
		
	It remains to consider the case that \( v \) is not among the first \( \tmax \) of label \( \alpha \) according to \( \getorder \).
	Then, we rely on that \( \getorder \) is nice to \( f \).
	We have that the \( (\tmax+1) \)st vertex of label \( \alpha \) has the first \( \tmax \) vertices of label \( \beta \) ordered before.
	Especially, \( v \) has \( \tmax \) many vertices of label \( \beta \) ordered before.
	This is equal to the amount of \( x \in [|A|] \) of label \( \beta \) with \( x < \deactt(v) = \alpha \), when capped by \( \tmax \).
	}
	\seeFullProof{\ref{appendix:lemma:key}}.
\end{proof}


\newcommand{\lemmaAdding}{
	Graph  \( \GG( \addd f' ) \) is \( k \)-activating for state \( (A,\afo) \)
		if and only if
		\( A \) is nice to \( \addd f' \) and
		\( \GG(f') \) is \( k \)-activating for \( (A, \addd\afo) \).
}
\begin{lemma}
\label{lemma:adding}
	\lemmaAdding
\end{lemma}
\begin{proof}[\proofSketch]
	We assume that
	 the \( \ell \)-expression \( \addd f' \) is irredundant as mentioned in the preliminaries.
	Then, every edge between vertices of label \( \alpha \) and \( \beta \) is new to \( \GG(f') \) such that
	\( \incoming{\getorder}{\addd f'}(v) \; = \; \incoming{\getorder}{f'}(v) + \adde(v) \).
	For the forward direction, let \( \GG(\addd f) \) have global ordering \( \getorder \) that extends \( A \), is \( k \)-activating for state \( (A,\afo) \) and nice to every subexpression fo \( \addd f' \).
	It follows directly that \( A \) is nice to \( \addd A \).
	We in particular show that the same ordering \( \getorder \) is \( k \)-activating for the modified state \( (A,\addd \afo) \).
	That is, every non-target set vertex \( v \) has \( \incoming{\getorder}{f'}(v) \geq \thr(v) - (\addd\afo)(\deactt(v)) \).
	We can follow this result from our initial observation and by applying Lemma \ref{lemma:key}.
	The backward direction is similar.
	\seeFullProof{\ref{appenidx:lemma:adding}}.
\end{proof}

In case of a current expression \( f = f_1 \oplus f_2 \), we have to show how to recursively rely on the subexpressions \( f_1 \) and \( f_2 \),
	analogously for \( f = \renamee f' \), on subexpression \( f' \).
For both cases, vertices of label \( \beta \) potentially come from different sets of vertices.
	In case of a re-labeling form \( \alpha \) to \( \beta \), a vertex of label \( \beta \) possibly had label \( \alpha \) before or already had label \( \beta \).
	In case of a disjoint union of subgraphs, a vertex of label \( \beta \) (or any other label) can be from either subgraph \( \GG(f_1) \) or \( \GG(f_2) \).	
For our indirect referenced vertices of our local ordering \( A \), we do not know the true origin.
	Thus, we have to try all possible partitions of label \( \beta \) into labels \( \alpha \) and \( \beta \),
	respective all partitions of label \( \beta \) (and every other label) into either subgraph.
As the possible local orderings \( A \) are bounded by our parameters, also the possible partitions are bounded by our parameters.

\flexdef{States for \( f_1 \oplus f_2 \) and \( \renamee f' \)}{
(1)
A state \( (A,\afo) \) of graph \( \GG(f) \) \define{completes} a state \( (A^\inc, \afo^\inc) \) if \( A \) is complete,
and removing from \( A \), for every label \( \alpha \), the last tuples of label \( \alpha \) from \( A \) until as many as in \( A^\inc \) remain, results in \( A^\inc \);
	and \( \afo: [|A|] \cup \labset(f) \to [0,\tmax], \) maps \( x \) to
	\( \afo^\inc(x) \), if defined for \( x \), and otherwise to \( \afo^\inc \big( \lab(A[x]) \big) \).
	
	(2)
	Let \( (f_1 \oplus f_2) \) be an \( \ell \)-expression.
	Then, \( \setjoin \) is the family of every pair of states \( \big( (A_1,\afo_1), (A_2, \afo_2) \big) \) that complete the possible incomplete states \( (A_1^\inc, \afo_1^\inc) \) and \( (A_2^\inc, \afo_2^\inc) \) that can be constructed as follows.
	Start with states \( (A_1^\inc, \afo_1^\inc) \), \( (A_2^\inc, \afo_2^\inc) \) where \( A_1^\inc = A_2^\inc = () \) and,
		for every label \( \alpha \), we have \( \afo_i^\inc(\alpha) = \afo(\alpha) \).
	For position \( j \), beginning from \( 1 \) to \( |A| \), add \( A[j] \) to the end of either list \( A_i^\inc \in \{A_1^\inc, A_2^\inc\} \) where possible.
	For position \( j \in [|A|] \), tuple \( A[j] \) is added to list \( A^\inc_i \), and let \( j' \) be the position of \( A[j] \) in \( A^\inc_i \).
		Then, let \( \afo_i^\inc(j') := \afo(j) \).
		
	(3)
	Let \( (A,\afo) \) be a state of \( \GG( \renamee f' ) \).
	Then, \( \setlabel \) is the family of every state \( (A', \afo') \) that completes a state \( (A^\inc,\afo^\inc) \) that can be constructed as follows.
	Re-label \( s \in [0, \tmaxx{\alpha}( f' )]  \) many tuples of \( A \) of label \( \beta \) to \( \alpha \), while at most \( \tmaxx{\beta}(f') \) of label \( \beta \) remain, resulting in \( A^\inc \).
	Let \( \afo^\inc \) be defined as \( \afo \) but where \( \afo^\inc( \alpha) = \afo(\beta) \).
}

\newcommand{\lemmaSplit}{
	Graph \( \GG( f_1 \oplus f_2 ) \) is \( k \)-activating for state \( (A,\afo) \)
		if and only if
		there are states \( \big( (A_1,\afo_1),(A_2,\afo_2) \big) \in \setjoin \) 
		and partition \( k_1 + k_2 = k \) such that,
		for \( i \in \{1,2\} \),
		graph \( \GG(f_i) \) is \( k_i \)-activating for \( (A_i, \afo_i) \).
}
\begin{lemma}
\label{lemma:split}
	\lemmaSplit
\end{lemma}
	\seeProof{\ref{appendix:lemma:split}}

\newcommand{\lemmaRelabel}{
	\comment{ 
	Let \( \renamee f' \) be an \( \ell \)-expression, and \( (A,\afo) \) a state of \( \GG( \renamee f') \).
	Then, graph \( \GG( \renamee f' ) \) is \( k \)-activating for state \( (A,\afo) \)
	if and only if
	there  is a state \( (A',\afo') \in \setlabel \) such that
	\( \GG(f') \) is \( k \)-activating for \( (A',\afo') \).
	}
	Graph \( \GG( \renamee f' ) \) is \( k \)-activating for state \( (A,\afo) \)
	if and only if
	there  is a state \( (A',\afo') \in \setlabel \) such that
	\( \GG(f') \) is \( k \)-activating for \( (A',\afo') \)
}
\begin{lemma}
\label{lemma:relabel}
	\lemmaRelabel
\end{lemma}
	\seeProof{\ref{appendix:lemma:relabel}}

Finally, we can show our main theorem, which was stated in the introduction.

\begin{theorem}[Theorem \ref{lemma:main:introduction} restated] 
	\lemmaMain
\end{theorem}
\begin{proof} 
	The minimal size of a target set is the minimal \( k \) of all local orderings \( A \) of \( \GG(f) \) such that \( \GG(f) \) is \( k \)-activating for \( (A,\nullf) \),
		as seen in Lemma \ref{lemma:translate:to:state}.

	Our algorithm computes the minimal \( k \) for possibly each subexpression \( f' \) of \( f \) and state \( (A,\afo) \) of \( \GG(f') \), in the fashion of dynamic programming.
	The minimum for a subexpression \( f' \) and state \( (A,\afo) \) of \( \GG(f') \) is remembered for future queries.
	There are at most \( ( \ell t )^{\ell t} \) possible local orderings \( A \) for a subgraph \( \GG(f') \).
	And there are at most \( t^{ \ell t + \ell } \) possible activations from outside \( \afo: [|A|] \cup \labset(f) \to [0,\tmax] \).
	Thus, there are at most \( ( \ell t )^{\ell t} \cdot t^{ \ell t + \ell } \) different states for a fixed subexpression.
	Further, every computation is the minimum of at most \( (\ell t )^{2 \ell t} \) entries (an upper bound is guessing \( A_1,A_2 \) respectively \( A' \) from scratch),
		and the minimum can be found in linear time.
	Therefore, the algorithm runs in time \( \Oh( ( \ell t )^{\ell t} \cdot t^{ \ell t + \ell } \cdot (\ell t )^{2 \ell t} ) \cdot |f| = \Oh( \ell^{3 \ell t} \cdot t^{\ell (4t+1) } \cdot |f| ) \).
	If \( (A,\afo) \) is not a correct state for \( \GG(f') \), set its minimum to \( \infty \).
	
	If \( f \) contains only one operation, then \( f = v(\alpha) \) 
	and the only possible global ordering is \( \getorder: \{v\} \to \{1\} \).
	Graph \( \GG(f) \) is at least \( 1 \)-activating, and possibly \( 0 \)-activating if \( \thr(v) \geq \thr(v) - \afo(1) \).
		Answer accordingly in time \( \Oh(1) \).

	Otherwise, if \( f \) consists of more than one operation, we have either of the recursive cases that \( f \) is \( \addd f' \), \( f_1 \oplus f_2 \) or \( \renamee f' \). 
	According to Lemma \ref{lemma:split} and \ref{lemma:relabel} respectively,
		graph \( \GG(f_1 \oplus f_2) \) is \( k \)-activating for state \( (A, \afo) \)  if and only if
			there is a pair of states \( \big( (A_1,\afo_1),(A_2,\afo_2) \big) \in \setjoin \) and partition \( k_1 + k_2 = k \) such that, for \( i \in \{1,2\} \), the graph \( \GG(f_i) \) is \( k_i \)-activating for \( (A_i,\afo_i) \);
		and graph \( \GG(f\renamee f') \) is \( k \)-activating if and only if
			there is a state \( (A',\afo') \in \setlabel \) such that \( \GG(f') \) is  \( k \)-activating for \( (A',\afo') \).
	Therefore, in those two cases we can recursively obtain a minimum size of a target set by querying for the according subgraphs \( \GG(f'), \GG(f_1), \GG(f_2) \) and states \( \big( (A_1,\afo_1), \allowbreak (A_2,\afo_2) \big) \in \setjoin \) and \( (A',\afo') \in \setlabel \), respectively.
	In case of \( f = f_1 \oplus f_2 \) the minimum size of a target set is the minimum of the sum of the minimum sizes for \( f_1 \) and \( f_2 \).
	For \( f = \renamee f \) the minimum size is equal to the minimum for \( f' \).
	
	According to Lemma \ref{lemma:adding}, graph \( \GG(\addd f') \) is \( k \)-activating for state \( (A,\afo) \)
		if and only if \( A \) is nice \( \addd f' \) and graph \( \GG(\addd f') \) is \( k \)-activating for state \( (A, \addd \afo) \).
	Thus, in case of that \( A \) is not nice to \( f \) we can discard the current computation for a minimal size of a target set for the graph \( \GG(\addd f') \) and state \( (A, \afo) \).
	Otherwise, the minimum size of a target set is equal to the minimum size of subgraph \( \GG(f) \) with state \( (A,\addd \afo) \).
\end{proof}

\section{Conclusion}

\label{section:not:main}

\label{section:conclusion}

In this work, we gave an \fpt-algorithm for \problemTSS for the combined parameters clique-width and \maximum threshold.
This result generalizes all previous \fpt-membership results of \problemTSS with constant thresholds.
It would be interesting to explore the whole dichotomy of constant \problemTSS for common structural parameters.
Is there a different dichotomy when the \maximum threshold is a parameter and not a constant?

\bibliographystyle{plain}
\bibliography{lit}

\arxiv{
\newpage

\appendix

\section{Omitted Proofs}

\subsection{Proof of Lemma \ref{lemma:fix}}
\label{appendix:lemma:fix}

\begin{lemma}[Lemma \ref{lemma:fix} restated]
	\lemmaFix
\end{lemma}
\begin{proof}
	Let \( \tmax := \tmax(f) \).
	We modify \( \getorder \) such that it is nice for \( f \) while still being \( k \)-activating for graph \( \GG(f) \).
	Let \( v_\alpha[1], v_\alpha[2], \dots \in V(f) \) and \( v_\beta[1], v_\beta[2], \dots \in V(f) \) be the vertices of label \( \alpha \) respectively label \( \beta \) ordered ascending according to \( \getorder \).
	Top-down for every subexpression \( \addd f \) we assure that (1) \( \getorder( {v_\alpha}[\tmax \! + \! 1] ) > \getorder( {v_{\beta}}[\tmaxx{\beta}] ) \), if defined, and vice-versa that (2) \( \getorder( {v_\beta}[\tmax+1] ) > \getorder( {v_{\alpha}}[\tmaxx{\alpha}] ) \), if defined.
	We begin to show how to locally fix such a violation.
	
	As the two conditions (1) and (2) are symmetric and not both can be false for the same subexpression, it suffices to consider that (\(\neg\)1) \( \getorder( {v_\alpha}[\tmax+1] ) =: i < \getorder( {v_{\beta}}[\tmaxx{\beta}] ) \).
	In that case, move the ordering of the at step \( i \) not yet occurred vertices of \( v_\beta[1],\dots, v_\beta[\tmaxx{\beta}] \) between position \( i-1 \) and \( i \).
	Let there be \( j \in \N \) many not yet occurred vertices.
	Then, our modification of \( \getorder \), puts the vetices \( v_\beta[\tmaxx{\beta} - (j \! - \! 1) ],\dots, v_\beta[\tmaxx{\beta}] \) to position \( i, \dots, i+ (j \! - \! 1) \), and delays the ordering of the following vertices \( \vertify(i), \vertify(i\!+\!1), \dots \) by \( j \) steps.
	
	In the following, we show that this modificated \( \getorder \) is still \( k \)-activating for graph \( \GG(f) \).
	Further, we observe that such a modification does not introduce a violation for added edges by an outer operation, those we already visited.
	Therefore, by recursively visiting every edge operation from top-down, we alter \( \getorder \) such that it is nice to \( f \),
		which proves the existance of a global ordering \( \getorder' \) that is activating for \( \GG(f) \) and is nice to \( f \).
	
	We claim that by this local modification, the altered global ordering \( \getorder' \) is still activating.
	For every vertex \( v \in V(f) \) that is delayed, we have that all neighbors ordered previously are preserved and assure the activation of \( v \).
	It remains to consider the vertices \( v_\beta[\tmax - (j \! - \! 1)],\dots, v_\beta[\tmax] \) whose position was shifted forward to \( i = \getorder( {v_\alpha}[\tmax\!+\!1] ) \) and following.
	However, as the ordering for positions \( 1, \dots, (i \! - \! 1) \) is unchanged, the vertices \( v_\alpha[1], \dots, v_\alpha[\tmax] \) still are ordered within positions \( 1, \dots, (i \! - \! 1) \).
	This means, that every forward shifted vertex of \( v_\beta[\tmax - (j \! - \! 1)],\dots, v_\beta[\tmax] \) still has at least \( \tmax \) many neighbors ordered before.
	As the \maximum threshold is \( \tmax \), it can get active.
	
	We claim that this modification does not cause violations at already seen expressions \( (\addda \dots \addd f') \) in our top-down approach.
		That means \( \getorder( v_{\alpha'}[\tmax+1] ) > \getorder( {v_{\beta'}}[\tmaxx{\beta'}] ) \) and \( \getorder( v_{\beta'}[\tmax+1] ) > \getorder( {v_{\alpha'}}[\tmaxx{\alpha'}] ) \) is still true. 
	Only \( \beta \) labeled vertices are moved forward, but not the potentially violating vertex \( {v_\beta}[\tmax \! + \! 1] \).
	Thus, especially the condition \( \getorder( {v_\beta}[\tmax \! + \! 1] ) > \getorder( {v_{\alpha}}[\tmaxx{\alpha}] ) \) is still true.
	Every outer operation with a current subexpression \( (\addda \dots \addd f') \) from our top-down approach, adds edges between vertex sets \( V_{\alpha'} \) of label \( \alpha' \) and \( V_{\beta'} \) of label \( \beta' \).
		Potentially, \( \beta \) has been re-labeled to \( \beta' \).
	However, as our \( v_{\beta}[1], v_\beta[2], \dots \) vertices from the nested expression \( \addd f' \) have pairwise equal label in every outer subexpression, we have that \( \{ v_{\beta}[1], v_\beta[2], \dots \} \subseteq V_{\beta'} \).
	That means, also for \( \addda \) we do not cause that \( \getorder'( {v_{\beta'}}[\tmax \! + \! 1] ) > \getorder'( {v_{\alpha'}}[\tmaxx{\alpha'}] ) \) as desired.
\end{proof}

\subsection{Proof of Lemma \ref{lemma:key}}
\label{appendix:lemma:key}

\begin{lemma}[Lemma \ref{lemma:key} restated]
	\lemmaKey
\end{lemma}
\begin{proof}
	For simplicity, let \( V := V(f) \). 
	For vertices \( v \in V \) of label \( \gamma \neq \{\alpha,\beta\} \), no edges are added and we have that \( \adde(v) = 0 = \addvalue{ \deactt(v) } \),
		which directly implies the statement.
	Otherwise \( v \) has either label \( \alpha \) or \( \beta \).
	As the two cases of are symmetric, let us only consider the case that \( v \) has label \( \alpha \).	
	Then, the terms \( \adde(v) \) and \( \addvalue{ \deactt(v) } \) simplify to
	\begin{align*}
		&\adde(v) \; = \; | \bigset{ u \in V }{ \getorder(u) < \getorder(v), \; \lab(u)= \beta } |,
		\\&\addvalue{ \deactt(v) } \; = \; |\bigset{ x \in [|A|] }{ x< \deactt(v), \; \lab(x) = \beta }|
		.
	\end{align*}
	We show that \( \min\{ \tmax, \; \adde(v) \} = \addvalue{ \deactt(v) } \), which implies the statement.
	\begin{align*}
		 \min\{ \tmax, \; \adde(v) \}
		\; = \; \min\!\big\{ \tmax, \; |\bigset{ u \in V }{ \getorder(u) < \getorder(v), \; \lab(u) = \beta }| \big\}
	\end{align*}
	as the first \( \tmax \) many vertices of label \( \beta \) are among \( \condense(\getorder) =: a_1, \dots, a_{|A|} \),
	\begin{align*}
	\label{test}
		=& \; \min\!\big\{ \tmax, \; |\bigset{ x \in [|A|] }{ \getorder( a_x ) < \getorder(v), \; \lab(x) = \beta }| \;  \big\} 
		\\
		=& \; \min\!\big\{ \tmax, \; |\bigset{ x \in [|A|] }{
				x<\deactt(v), \; \lab(x) = \beta }| \;  \big\} \;
				 =  \; \addvalue{ \deactt(v) } 
				,
	\end{align*}
	For the second las equality, we show in the following that, the number of positions \( x \in [|A|] \) of label \( \beta \), where \( \getorder( a_x ) < \getorder(v) \) is equal to the number of positions where \( x<\deactt(v) \), given that we cap the numbers by \( \tmax \).
	This, then finishes our proof.
	
	We distinguish whether \( v \) occurs in \( a_1, \dots, a_{|A|} \),  formally if there is a \( y \in [|A|] \) such that \( v = a_y \).
	Assume that \( v = a_i \) for some \( i \in [|A|] \).
	Then, we have that \( \deactt(v) = y \).
	Thus, for every \( x \in [|A|]  \), we have \( \getorder(a_x) < \getorder(v) = \getorder(a_x) \) if and only if \( x < \deactt(v) = y \).
	Otherwise, for that \( v \) does not occur in \( a_1, \dots, a_{|A|} \), we proof the forward an backward containment in the following.
	
	(\( \supseteq \))
	Assume that \( v \), of label \( \alpha \), does not occur in \( a_1, \dots, a_{|A|} \).
	Then, there are at least \( \tmax \) many vertices of label \( \alpha \) ordered before \( v \) in the current subgraph \( \GG(f) \).
	Let \( v_\alpha[1], v_\alpha[2] \dots \) be the ordering of \( \alpha \) labeled vertices in \( \GG(f) \),
		analogously let \( v_\beta[1], v_\beta[2] \dots \) be the ordering of \( \beta \) labeled vertices.
	Because the local ordering \( A \) is nice to \( f \), we have that \( \getorder(v) \geq \getorder( {v_\alpha}[\tmax+1] ) > \getorder( {v_{\beta}}[\tmax] ) \).
	Thus, there are at least \( \tmax \) many  \( x \in [|A|] \) of label \( \beta \), fo which vertex \( a_x \) is ordered before \( v \).
	
	(\( \subseteq \))
	Having \( x < \deactt(v) \) for all  \( x \in [|A|] \) implies that \( \deactt(v) = \alpha \) and that \( v \) is not among the first active \( \alpha \) labeled vertices \( v_\alpha[1], \dots, v_\alpha[\tmax] \).
	Because global ordering \( \getorder \) is nice to \( f \), we have that \( \getorder(v) \geq \getorder( {v_\beta}[\tmax+1] ) > \getorder( {v_{\alpha}}[\tmax] ) \).
	This means that before the position \( \getorder( v ) \) there are at least the previous ordered \( \tmax \) many neighbors \( v_{\alpha}[1], \dots, v_{\alpha}[\tmax] \).
\end{proof}

\subsection{Proof of Lemma \ref{lemma:adding}}

\label{appenidx:lemma:adding}
\begin{lemma}[Lemma \ref{lemma:adding} restated]
	\lemmaAdding
\end{lemma}
\begin{proof}
	As mentioned in the preliminaries, we assume that the \( \ell \)-expression \( \addd f' \) is irredundant.
	That means, every edge between vertices of label \( \alpha \) and \( \beta \) is new to \( \GG(f') \) such that
			\( \incoming{\getorder}{\addd f'}(v) \; = \; \incoming{\getorder}{f'}(v) + \adde(v) \).
	
	\forward
	Let \( \GG(\addd f') \) be \( k \)-activating for state \( (A,\afo) \).
	That is there is a global ordering \( \getorder \) that extends \( A \), is \( k \)-activating for state \( (A,\afo) \) and is nice to every subexpression of \( \addd' f \).
	As \( \getorder \) is nice to every subexpression of \( \addd f \), especially \( \getorder \) is nice to every subexpression of \( f' \).
	Moreover, \( \getorder \) is also nice to \( f \) and \( A \) extends \( \getorder \) which implies that \( A \) is nice to \( \addd f' \) and we can apply Lemma \ref{lemma:key} on \( \getorder \) and \( A \).
	Since \( V(f') \) and \( V(\addd f') \) have equivalent labeling, \( \getorder \) also extends \( A \) for graph \( \GG(f') \).
	It remains to show that \( \getorder \) is \( k \)-activating for state \( (A, \addd \afo) \) on graph \( \GG(f') \).
	For every vertex \( v \in V(f') \) but for \( k \) exceptions, as \( \thr(v) \in [0,\tmax] \), we have that 
	\begin{align*}
		 \incoming{\getorder}{f'}(v) \; \geq \; & \max\big\{ 0, \; \incoming{\getorder}{\addd f'}(v) - \adde(v) \big\} \\
		 = \; & \max\big\{ 0, \; \thr(v) - \afoa{ v } - \adde(v) \big\} \\
		 \geq \; & \max\big\{ \thr(v)- \tmax, \; \thr(v) - \afoa{ v } - \adde(v) \big\} \\
		 \geq \; & \thr(v) - \min\big\{ \tmax, \; \afoa{ v } + \adde(v) \big\} \\
		 \overset{\text{L.~\ref{lemma:key}}}{=} & \thr(v) - (\addd\afo)(\deactt(v)) .
	\end{align*}
	Thus, graph \( \GG(f') \) has global ordering \( \getorder \), that extends \( A \), is \( k \)-activating for \( (A,\addd \afo) \) and thatis nice for every subexpression of \( f' \).
	Therefore, graph \( \GG(f') \) is  \( k \)-activating for state \( (A,\addd \afo) \).
	As seen before \( A \) is nice to \( \addd f \).
	
	\backward
	Let \( \GG(f') \) be \( k \)-activating for state \( (A,\addd\afo) \) and \( A \) be nice to \( \addd f \).
	The former implies that \( \GG(f') \) has global ordering \( \getorder \) that extends \( A \), is \( k \)-activating for \( (A,\addd\afo) \) and that is nice to every subexpression of \( f' \).
	As \( A \) be nice to \( \addd f \), we can apply Lemma \ref{lemma:key} on \( \getorder \) and \( A \).
	Because \( \getorder \) extends \( A \), which is nice to \( \addd f' \), also \( \getorder \) is nice to \( \addd f' \).
		Hence, \( \getorder \) is nice to all subexpressions of \( \addd f' \).
	Since the vertices \( V(f') \) and \( V(\addd f') \) have equivalent labels, \( \getorder \) also extends \( A \) for graph \( \GG(\addd f') \).
	It remains to show hat \( \getorder \) is \( k \)-activating for \( (A,\addd\afo) \) on graph \( \GG(\addd f') \).
	For every vertex \( v \in V(\addd f') \) but for \( k \) exceptions, we have that 
	\begin{align*}
		 \incoming{\getorder}{\addd f'}(v) \; = \; & \incoming{\getorder}{f'}(v) + \adde(v) \\
		 \geq \; & \thr(v) - (\addd\afo)(\deactt(v)) + \adde(v) \\
		  \overset{\text{L.~\ref{lemma:key}}}{=} & \thr(v) - \min\big\{ \tmax, \; \afoa{ v } + \adde(v) \big\} + \adde(v) \\
		 \geq \; & \thr(v) - \big( \afoa{ v } + \adde(v) \big) + \adde(v) \\
		 = \; & \thr(v) - \afoa{v} ,
	\end{align*}
	Thus, graph \( \GG(\addd f') \) has global ordering \( \getorder \), that extends \( A \), is \( k \)-activating for \( (A, \afo) \) and nice for every subexpression of \( \addd f' \).
	Therefore, graph \( \GG(\addd f') \) is  \( k \)-activating for state \( (A,\afo) \).
\end{proof}

\subsection{Proof of Lemma \ref{lemma:translate:to:state}, \ref{lemma:split} and \ref{lemma:relabel}}

We first introduce a tool to complete states.

\begin{lemma}
\label{lemma:extend:A}
	Let \( f \) be an \( \ell \)-expression, and \( (A^\inc,\afo^\inc) \) a possibly not complete state of \( \GG(f) \).
	Let \( \GG(f) \) be \( k \)-activating for state \( (A^\inc,\afo^\inc) \).
	Then, there is state \( (A,\afo) \) of \( \GG(f) \) completing \( (A^\inc,\afo^\inc) \) such that \( \GG(f) \) is \( k \)-activating for \( (A,\afo) \).
\end{lemma}
\begin{proof}
	Let \( \GG := \GG(f) \) have global ordering \( \getorder: V(\GG) \to [|V(\GG)|] \) that extends \( A^\inc \), is \( k \)-activating for \( (A^\inc,\afo^\inc) \) and is nice to every subexpression of \( f \).
	Consider the ordered list of vertices and \( \vertify(1), \dots, \vertify(|V(\GG)|) \), which we underline or mark with a star as follows.
	For every label \( \alpha \), underline the first \( \tamount{\alpha}(\GG) \) occurrences of \( \alpha \) labeled vertices.
	Then, for every label \( \alpha \), mark the first \( |\set{ i }{ \lab(A^\inc[i]) = \alpha }| \) occurrences of \( \alpha \) labeled vertices with a star.
	Since \( |\set{ i }{ \lab(A^\inc[i]) = \alpha }| \leq \tamount{\alpha}(\GG) \), every vertex marked with a star is also underlined.
	Further, the list of underlined vertices is equal to \( \condense(\getorder) \) while the list of vertices marked with a star is equal to \( \condense(\getorder, A^\inc) \).
	Let \( A \) be the list of tuples of label and threshold of \( \condense(\getorder) \).
	Then, \( A \) is a complete local ordering and extends \( \getorder \).
	Further, by deleting, for every label \( \alpha \), the last occurring tuples of label \( \alpha \) until as many as \( |\set{ i }{ \lab(A^\inc[i]) = \alpha }| \) of label \( \alpha \) remain,
		we remove the tuples of that are underlined but without a star.
	Therefore, \( A \) extends \( A^\inc \).
	Let \( \afo \) uniquely be such that \( (A,\afo) \) completes \( (A^\inc, \afo^\inc) \).
	Then, \( \getorder \) extends \( A^\inc \) and is \( k \)-activating for \( (A,\afo) \) and is nice to every subexpression of \( f \).
	Thus, \( \GG(f) \) is \( k \)-activating for \( (A,\afo) \).
\end{proof}

	\label{appendix:lemma:translate:to:state}
\begin{lemma}[Lemma \ref{lemma:translate:to:state} restated]
		\lemmaTranslateToState
\end{lemma}
\begin{proof}
	\forward
	Let \( \GG(f) \) have global ordering \( \getorder \) be \( k \)-activating for graph \( \GG(f) \).
	Then, according to Lemma \ref{lemma:fix} there also is a global ordering \( \getorder' \) such that \( \getorder' \) is \( k \)-activating for \( \GG \) and is nice to every subexpression of \( f \).	
	Thus, there is a \( k \)-vertex set \( S \subseteq V(f) \) such that for every vertex \( v \in V(f) \setminus S \) we have 
	\[
		\incoming{\getorder'}{f}(v) \; \geq \; \thr(v) = \thr(v) - \nullf'( \deactt(v) ) ,
	\]
	where \( \nullf': [\ell] \to \{0\} \).
	Therefore, graph \( \GG(f) \) has a global ordering \( \getorder' \) that extends the empty list \( () \), is \( k \)-activating for \( ( (), \nullf') \)
		and nice to every subexpression of \( f \).
	As seen in Lemma \ref{lemma:extend:A} there is a complete state \( (A,\afo) \) that extends \( ((),\nullf') \) such that \( A \) is nice to \( f \) and \( \GG(f) \) is \( k \)-activating for \( (A,\afo) \).
	Extending \( \nullf' \) results in a an activation from outside \( \nullf: [|A|] \cup \labset \to \{0\} \).
	Thus, we have \( \afo = \nullf \) and that \( \GG(f) \) is \( k \)-activating for state \( (A,\nullf) \).
	
	\backward
	Let \( \GG(f) \) be \( k \)-activating for state \( (A,\nullf) \),
		which means that there is a global ordering \( \getorder \) that is \( k \)-activating for \( (A,\nullf) \).
	Then there is a \( k \)-vertex set \( S \) such that for every vertex \( v \in V(f) \setminus S \) we have that
	\(
		\incoming{\getorder}{f}(v)
		 \geq  \thr(v) - \nullf( \deactt(v) ) \; = \; \thr(v) .
	\)
	Therefore, \( \getorder \) is \( k \)-activating for \( \GG(f) \).
\end{proof}

\label{appendix:lemma:split}
\begin{lemma}[Lemma \ref{lemma:split} restated]
	\lemmaSplit
\end{lemma}
\begin{proof}
	We begin with two observations (1) and (2), which we use to in both directions of the proof \forward and \backward.
	
	(1) Assume that \( \GG(f_1 \oplus f_2) \) has global ordering \( \getorder \), and for \( i \in \{1,2\} \) graph \( \GG(f_i) \) has as global ordering \( \getorder_i \) where the relative ordering of vertices \( V(f_i) \) is equal.
	Then, we have that \( \incoming{\getorder_i}{(f_i)}(v) = \incoming{\getorder}{f_1 \oplus f_2}(v) \)
		because there are no edges between \( V(f_1) \) and \( V(f_2) \) in \( \GG(f_1 \oplus f_2) \), and \( \getorder \) preserves the relative ordering of the vertices of \( \getorder_i \).
	
	(2) Assume that \( \GG(f_1 \oplus f_2) \) has global ordering \( \getorder \), and for \( i \in \{1,2\} \) graph \( \GG(f_i) \) has global ordering \( \getorder_i \) and local ordering \( A^\inc_i \) such that for every vertex \( v \in V(f_i) \) we have that \( v \in \condense(\getorder) \) if and only if \( v \in \condense(\getorder,A^\inc_i) \).
		We show that, for every vertex \( v \in V(f_i) \), we have that \( \afo(\deactt(v)) = \afo_i^\inc(\deacttt{A^\inc_i}{\getorder_i}(v)) \).
	Consider the case that \( v \in \condense(A) \).
	Let \( \condense(\getorder) = a_1,a_2,\dots \) and \( \condense(\getorder) = a^i_1, a^i_2, \dots \).
	Then, there is a position \( j \in [|A|] \) such that \( v = a_j = a^i_{j'} \).
	It follows the desired equality \( \afo(\deactt(v)) = \afo(a_j) = \afo^\inc_i(a^i_{j'}) = \afo_i^\inc(\deacttt{A^\inc_i}{\getorder_i}(v)) \).
	Otherwise, it is the case that \( v \notin \condense(\getorder) \).
	Then, also \( v \notin \condense(\getorder,A^\inc_i) \).
	Thus, there is a label \( \alpha \in \labset(f_i) \) such that we have \( \afo(\deactt(v)) = \afo(\alpha) = \afo^\inc(\alpha) = \afo_i^\inc(\deacttt{A^\inc_i}{\getorder_i}(v)) \).
	
	\forward
	Let \( \GG(f_1 \oplus f_2) \) be \( k \) activating for state \( (A,\afo) \).
	Then, there is a global activating ordering \( \getorder \) that extends \( A \), is \( k \)-activating for \( (A,\afo) \) and is nice to every subexpression of \( (f_1 \oplus f_2) \).
	Let \( S \subseteq V(f_1 \oplus f_2) \) be the \( k \)-vertex set such that for every other vertex \( v \in V(f_1 \oplus f_2) \setminus S \) we have that \( \incoming{\getorder}{f_1 \oplus f_2}(v) \geq \thr(v) - \afo(\deactt(v)) \).
	For \( i \in \{1,2\} \), let \( k_i \) be equal to \( |S \cap V(\GG_i)| \), which implies that \( k_1 + k_2 = k \).
	Let \( \condense( \getorder ) = a_1, \dots, a_s \).
	For \( j \in [|A|] \), put the tuple \( A[j] \) into that list \( A_i^\inc \) where its according vertex \( a_j \in V(f_i) \) is from.
	For \( i \in \{1,2\} \), we show the that \( \GG(f_i) \) is \( k_i \)-activating for state \( (A_i^\inc, \afo_i^\inc) \) in the following.
	Let global ordering \( \getorder_i: V(f_i) \to [|V(f_i)|] \) map every \( v \in V(f_i) \) according to its position in \( \getorder \), which is \( \getorder_i(v) := |\set{ u \in V(f_i) }{ \getorder(u) \leq \getorder(v) }| \).
	
	We show that global ordering \( \getorder_i \) extends \( A^\inc_i \),
		which means that the list of tuples of label and threshold of \( \condense(\getorder_i) \) is equal to \( A \).
	We have that \( A^\inc_i = \set{ A[z] }{ z \in [s], \; a_z \in V(f_i) } \) as the above definition implies.
	Then, list \( \condense( \getorder_i, A^\inc_i ) \) is the list of vertices of, for every label \( \alpha \in \labset(f_i) \), the first \( \tmaxx{\alpha}(f_1 \oplus f_2) \) vertices of label \( \alpha \) that are in \( \GG(f_i) \).
	Therefore, list \( \condense( \getorder_i, A^\inc_i ) \) is the list \( \condense( \getorder ) \) restricted to vertices \( V(f_i) \), that is \( \condense( \getorder_i, A^\inc_i )\allowbreak =\allowbreak \set{ a_z }{ z \in [s], a_z \in V(f_i) } \).
	Then, it follows that the list of tuples of label and threshold of \( \condense( \getorder_i, A^\inc_i ) \) is equal to \( A^\inc_i \).
	Thus, \( \getorder_i \) extends \( A^\inc_i \).
	Moreover, for every vertex \( v \in V(f_i) \) we have that \( v \in \condense(\getorder) \) if and only if \( v \in \condense(\getorder,A^\inc_i) \).
	Hence, for every vertex \( v \in V(f_i) \), we have that \( \afo(\deactt(v)) = \afo_i^\inc(\deacttt{A^\inc_i}{\getorder_i}(v)) \) as seen in observation (2).
	
	We have that \( \getorder_i \) preserves the relative ordering of vertices of \( \getorder \) that are in subgraph \( \GG(f_i) \).
	Since \( \getorder \) is nice to every subexpression of \( (f_1 \oplus f_2) \), also \( \getorder_i \) is nice to every subexpression of \( f_i \).
	
	It remains to show that \( \getorder_i \) is \( k_i \)-activating for state \( (A^\inc_i,\afo^\inc_i) \).
	Let \( S \subseteq V(f_1 \oplus f_2) \) be the \( k \)-vertex set such that for every other vertex \( v \in V(f_1 \oplus f_2) \setminus S \) we have that \( \incoming{\getorder}{f_1 \oplus f_2}(v) \geq \thr(v) - \afo(\deactt(v)) \).
	Then, for every vertex \( v \in V(f_i) \setminus S \) it follows that 
	\[
		 \incoming{\getorder_i}{(f_i)}(v) \; 
			 \overset{(1)}{=} \; \incoming{\getorder}{f_1 \oplus f_2}(v) 
			 \geq \; \thr(v) - \afo(\deactt(v)) 
			 \overset{(2)}{=} \; \thr(v) - \afo_i^\inc(\deacttt{A^\inc_i}{\getorder_i}(v))
			 .
	\]
	Thus, graph \( \GG(f_i) \) has global ordering \( \getorder_i \) that extends \( A_i^\inc \), is \( k_i \)-activating for \( (A^\inc_i, \afo^\inc_i) \) and nice to every subexpression of \( f_i \).
	Therefore, graph \(  \GG(f_i) \) is \( k_i \)-activating for  \( (A^\inc_i, \afo^\inc_i) \). 
	Finally, extend the possibly incomplete state \( (A_i^\inc, \afo_i^\inc) \) to a complete state \( (A_i, \afo_i) \) such that \( \GG(f_i) \) is \( k_i \)-activating for \( (A_i, \afo_i) \) as seen in Lemma \ref{lemma:extend:A}.
	
	\backward
	Let \( \big( (A_1,\afo_1), (A_2, \afo_2) \big) \in \setjoin \) and \( k_1,k_2 \in \N \) be such that, for \( i \in \{1,2\} \), graph \( \GG(f) \) is \( k_i \)-activating for state \( (A_i,\afo_i) \).
	That means, \( \GG(f_i) \) has global ordering \( \getorder_i \) that extends \( A_i \), is \( k \)-activating for \( (A_i,\afo_i) \) and is nice for every subexpression of \( f_i \).
	Let \( (A_i, \afo_i) \) extend \( (A_i^\inc,\afo_i^\inc) \) as seen in the construction.
	Then, specially \( \getorder_i \) extends \( A_i^\inc \), is \( k \)-activating for \( (A_i^\inc,\afo_i^\inc) \) and is nice for every subexpression of \( f_i \).
	
	We define the global ordering \( \getorder \).
	Following this we show that \( \getorder \) extends \( A \), is \( k \)-activating for \( (A,\afo) \) and nice to every subexpression of \( (f_1 \oplus f_2) \).
	Let mapping \( i: [|A|] \to \{1,2\} \) be such that, for all positions \( j \in [|A|] \), we have that it maps to the according subgraph, which is \( A[j] = A_{i(j)}[j'] \).
	For \( i \in \{1,2\} \), let \( \condense( A^\inc_i, \afo ) = a^i_1, \dots, a^i_{s^i} \).
	Then, for \( i \in \{1,2\} \) and \( j \in [|A_i|] \), let \( \nexte{i}(j) \) be the list of positions ranging from after the position of \( a_{j-1} \) to the position of \( a_j \) (or from position 1 if \( a_{j-1} \) does not exist).
	That is \( \nexte{i}(1) :=  \getorder(1), \dots,  \getorder(a_j^i) \) and \( \nexte{i}(j) :=  \getorder(a_{j-1}^i) + 1, \dots,  \getorder(a_j^i) \) for \( j \in \{2,\dots,|A|\} \).
	Let \( \getorder \) be such that its ordered vertices are \( \nexte{i(1)}( 1' ) , \dots, \nexte{i(|A|)}( |A|' ) \), which is the ordering that always adds the vertices not yet added up to the position of \( a_j \) from the ordering of the subgraph of \( a_j \).
	
	We show that \( A \) extends \( \getorder \), which means that \( A \) has tuples of label and threshold equal to \( \condense(\getorder) \).
	We have that \( A = A^\inc_{i(1)}[1'], \dots, A^\inc_{i(s)}[s'] \).
	For every label \( \alpha \in \labset(\GG) \), the number of tuples of label \( \alpha \) of \( A^\inc_1 \) and \( A^\inc_2 \) add up to the number of \( A \), in other words \( |\set{ z }{ \lab( A^\inc_1[z] ) = \alpha }| +  |\set{ z }{ \lab( A^\inc_2[z] ) = \alpha }| = \tmaxx{\alpha}(f_1 \oplus f_2) \).
	Thus, the list \( \condense(\getorder) = a^{i(1)}_{i'}, \dots, a^{i(|A|)}_{|A|'} \) has tuples of label and threshold equal to \( A = A^\inc_{i(1)}[1'], \dots, A^\inc_{i(s)}[s'] \).
	Moreover, for every vertex \( v \in V(f_i) \) we have that \( v \in \condense(\getorder) \) if and only if \( v \in \condense(\getorder,A^\inc_i) \).
		Hence, for every vertex \( v \in V(f_i) \), we have that \( \afo(\deactt(v)) = \afo_i^\inc(\deacttt{A^\inc_i}{\getorder_i}(v)) \) as seen in observation (2).
	
	We have that \( \getorder \) preserves the relative ordering of \( \getorder_1 \) and \( \getorder_2 \).
	Since, for \( i \in \{1,2\} \), global ordering \( \getorder_i \) is nice to every subexpression of \( f_i \) and \( \getorder \) is trivially nice to \( (f_1 \oplus f_2) \) itself,
		it follows that \( \getorder \) is nice to every subexpression of  \( (f_1 \oplus f_2) \).
	
	It remains to show that \( \getorder \) is \( k \)-activating for state \( (A,\afo) \).
	For \( i \in \{1,2\} \), let \( S_i \subseteq V(f_i) \) be the \( k_i \)-vertex set such that for every vertex \( v \in V(f_i) \setminus S_i \) we have that \( \incoming{\getorder_i}{(f_i)}(v) \geq \thr(v) - \afo_i^\inc(\deacttt{A^\inc_i}{\getorder_i}(v)) \).
	Then, for every vertex \( v \in V(f_1 \oplus f_2) \setminus (S_1 \cup S_2) \), there is an \( i \in \{1,2\} \) such that
	\[
		 \incoming{\getorder}{f_1 \oplus f_2}(v) \;	 		 
			\overset{(1)}{=} \; \incoming{\getorder_i}{(f_i)}(v) 
			\geq \; \thr(v) - \afo_i^\inc(\deacttt{A^\inc_i}{\getorder_i}(v))
			\overset{(2)}{=} \; \thr(v) - \afo(\deactt(v))
			.
	\]
	Thus, graph \( \GG(f) \) has global ordering \( \getorder \) that extends \( A \), is \( k \)-activating for \( (A, \afo) \) and nice to every subexpression of \( (f_1 \oplus f_2) \).
	Therefore, graph \( \GG(f_i) \) is \( k \)-activating for \( (A, \afo) \).
\end{proof}

\label{appendix:lemma:relabel}

\begin{lemma}[Lemma \ref{lemma:relabel} restated]
	\lemmaRelabel
\end{lemma}
\begin{proof}
	We begin with two observations (1) and (2), which we use to in both directions of the proof \forward and \backward.
	
	(1)
	Let \( \getorder \) be a global ordering of graph \( \GG(f) \) or of graph \( \GG(\renamee f) \).
	The number of predecessors of a vertex defined by \( \getorder \) does not depend on the labels.
	Thus, for every vertex \( v \in V(f') \) and for every vertex \( v \in V(\renamee f') \), we have that \( \incoming{\getorder}{f'}(v) = \incoming{\getorder}{\renamee f'}(v) \).
	
	(2)
	Let state \( (A,\afo) \) extend the possible incomplete state \( (A^\inc,\afo^\inc) \) as seen in the construction.
	Further, let \( \getorder \) be a global ordering.
	We show that, for every vertex \( v \in V(f') \), we have that \( \afo( \deactt(v) ) =\afo^\inc( \deact{A^\inc}(v)) \).
	If \( \deact{A^\inc}(v) = \alpha \), it follows that \( \afo( \deactt(v) ) = \afo( \beta ) = \afo^\inc( \alpha ) = \afo^\inc( \deact{A^\inc}(v)) \),
			and otherwise if \( \deact{A^\inc}(v) \neq \alpha \), there is an \( x \in [|A|] \cup \labset \setminus \{\alpha\} \) such that \( \afo( \deactt(v) ) = \afo( x ) = \afo^\inc( x ) =  \deact{A^\inc}(v) \).
	
	\forward
	Let \( \GG( \renamee f' ) \) be \( k \)-activating.
	Then, there is a global ordering \( \getorder \) that extends \( A \), is \( k \)-activating for \( (A,\afo) \) and is nice to every subexpression of \( \addd f' \).	
	Since \( \getorder \) is nice to every subexpression of \( \addd f' \), it especially is nice to every subexpression of \( f' \).
	Let \( \condense(\getorder) = a_1, \dots, a_{|A|} \).
	Let \( i_1 < \dots < i_s \) be the positions of \( A \) where \( a_{i_1},\dots,a_{i_s} \) have label \( \alpha \) in \( \GG(f') \).
	In particular, the positions \( a_{i_1},\dots,a_{i_s} \) have label \( \beta \) in \( G(\renamee f') \).
	Let \( A^\inc \) be the local ordering of \( \GG(f') \) equal to \( A \) but where the tuples \( A[i_1], \dots, A[i_s] \) are re-labeled from \( \beta \) to \( \alpha \).
	
	We show that \( \getorder \) extends \( A^\inc \) for graph \( \GG(f') \), which means that the tuples of label and threshold of \( \condense(\getorder) \) are equal to \( A^\inc \).
	Since \( \getorder \) already extends \( A \) for graph \( \GG(\renamee f') \), it suffices to show that the re-labeleing of vertices \( a_1, \dots, a_{|A|} \) matches the re-labeling of \( A \).
	The re-labeled vertices of \( A \) are \( a_{i_1}, \dots, a_{i_s} \), and they have the positions \( i_1, \dots, i_s \) where \( A \) is re-labeled.
	Thus, it follows that \( \getorder \) also extends \( A^\inc \) for graph \( \GG(f') \).
	
	It remains to show that \( \getorder \) is activating for state \( (A^\inc, \afo^\inc) \).
	Let \( S \subseteq V(\renamee f) \) be the \( k \)-vertex set such that for every other vertex \( v \in V(\renamee f) \setminus S \) we have that \( \incoming{\getorder}{\renamee f'}(v) \geq \thr(v) - \afo( \deactt(v) ) \).
	Then, for every vertex \( v \in V(f') \setminus S \), it follows that
	\[
		\incoming{\getorder}{f'}(v) \;
		\overset{(1)}{=} \; \incoming{\getorder}{\renamee f'}(v)
		\geq \; \thr(v) - \afo( \deactt(v) )
		\overset{(2)}{=} \; \thr(v) - \afo^\inc( \deact{A^\inc}(v) )
		.
	\]
	Thus, \( \getorder \) is a global ordering that extends \( A^\inc \), is \( k \)-activating for state \( (A^\inc, \afo^\inc) \) and nice to every subexpression of \( f' \).
	Therefore, graph \( \GG(f') \) is \( k \)-activating for graph \( (A^\inc,\afo^\inc) \).
	Finally, extend the possibly incomplete state \( (A^\inc, \afo^\inc) \) to a complete \( (A, \afo) \) such that \( \GG(f') \) is \( k \)-activating for \( (A, \afo) \) as seen in Lemma \ref{lemma:extend:A}.
		
	\backward
	Let state \( (A', \afo') \in \setlabel \) be such that graph \( \GG(f') \) is \( k \)-activating for \( (A', \afo') \).
	Let \( (A^\inc, \afo^\inc) \) be the possibly incomplete state from the construction of \( (A',\afo') \).
	Then, especially graph \( \GG(f') \) has global ordering \( \getorder \) that extends \( A^\inc \) and is \( k \)-activating for \( (A^\inc, \afo^\inc) \).
	Thus, \( \GG( f' ) \) has global ordering \( \getorder \) that extends \( A^\inc \), is \( k \)-activating for \( (A^\inc, \afo^\inc) \) and is nice to every subexpression of \( f' \).
	Let \( s \in [0,\tmax] \) be the number of tuples that have label \( \alpha \) in \( A^\inc \).
	Let \( i_1 <\dots < i_s \) be the positions in \( \getorder \) of vertices \( \condense(\getorder) \) that have label \( \alpha \) in \( A^\inc \). 
	
	We show that \( \getorder \) extends \( A \) for graph \( \GG( \renamee f') \), which means that the tuples of label and threshold of \( \condense(\getorder) \) are equal to \( A \).
	Since \( \getorder \) already extends \( A^\inc \) for graph \( \GG( f' ) \), it suffices to show that the re-labeling of vertices of vertices \( a_1, \dots, a_{|A|} \) matches the re-labeling of \( A^\inc \).
	The re-labeled vertices of \( A^\inc \) are \( a_{i_1}, \dots, a_{i_s} \), and they exactly have the positions \( i_1, \dots, i_s \) where \( A^\inc \) is re-labeled.
		Thus, it follows that \( \getorder \) also extends \( A \) for graph \( \GG(f') \).
	
	It remains to show that \( \getorder \) is  \( k \)-activating for state \( (A,\afo) \).
	Let \( S \) be the \( k \)-vertex set such that for every other vertex vertex \( v \in V(f') \setminus S \) we have that \( \incoming{\getorder}{f'}(v) \geq \thr(v) - \afo^\inc( \deact{A^\inc}(v) ) \).
	Then, for every vertex \( v \in V(\renamee f') \setminus S \), it follows that
	\[
		\incoming{\getorder}{\renamee f'}(v) \; \overset{(1)}{=} \; \incoming{\getorder}{f'}(v) \; \geq \; \thr(v) - \afo^\inc( \deact{A^\inc}(v) ) \; \overset{(2)}{=} \; \thr(v) - \afo( \deactt(v) )
		,
	\]
	and thus \( \getorder \) is a global ordering that extends \( A \), is \( k \)-activating for state \( (A, \afo) \) and nice to every subexpression of \( f' \).
	Therefore, graph \( \GG(f') \) is \( k \)-activating for \( (A,\afo) \).
\end{proof}

}{}

\end{document}